\documentclass[%
 aip,
jmp,
 amsmath,amssymb,
reprint, onecolumn 
]{revtex4-1}

\usepackage{graphicx}
\usepackage{dcolumn}
\usepackage{bm,mathrsfs,bbold,upgreek}

\usepackage[utf8]{inputenc}
\usepackage[T1]{fontenc}
\usepackage{mathptmx}
\usepackage{amsmath,amssymb,amsfonts}
\usepackage{etoolbox}

\makeatletter
\def\@email#1#2{%
 \endgroup
 \patchcmd{\titleblock@produce}
  {\frontmatter@RRAPformat}
  {\frontmatter@RRAPformat{\produce@RRAP{*#1\href{mailto:#2}{#2}}}\frontmatter@RRAPformat}
  {}{}
}%
\makeatother

\newtheorem{theorem}{Theorem}

\newtheorem{corollary}[theorem]{Corollary}

\newtheorem{proposition}[theorem]{Proposition}
\newtheorem{remark}[theorem]{Remark}

\newenvironment{proof}[1][Proof]{\noindent\textbf{#1.} }{\ \rule{0.5em}{0.5em}}
\begin{document}


\title{Quantum Filtering at Finite Temperature} 



\author{John Gough}
\email{jug@aber.ac.uk}
\affiliation{Aberystwyth University, SY23 3BZ, Wales, United Kingdom}

\date{\today}

\begin{abstract}
We pose and solve the problem of quantum filtering based on continuous-in-time quadrature measurements (homodyning) for the case where the quantum process is in a thermal state.
The standard construction of quantum filters involves the determination of the conditional expectation onto the von Neumann algebra generated by the measured observables - with the non-demolition principle telling us to restrict the domain (the observables to be estimated) to the commutant of the algebra. The finite-temperature case, however, has additional structure: we use the Araki-Woods representation for the measured quadratures, but the Tomita-Takesaki theory tells us that there exists a separate, commuting representation and therefore the commutant will have a richer structure than encountered in the Fock vacuum case. We apply this to the question of quantum trajectories to the Davies-Fulling-Unruh model. Here, the two representations are interpreted as the fields in the right and left Rindler wedges.
\end{abstract}


\maketitle 


\section{Introduction}
In this paper, we address the question of filtering when the environment is a bosonic input channel in a thermal state and we measure a quadature in its output channel after scattering from the system. The input is naturally described as a linear combination (Bogoliubov transformation) of two vacuum-state (i.e., Fock) inputs - this is an Araki-Woods representation (see Section 2). 

Quantum filtering is the theory of optimal estimation of the state of a quantum system based on indirect time-continuous measurements of its environment \cite{Belavkin,Slava}. An essential assumption is that the environment interacts with the system in a Markovian manner and this requires a filtration: that is, we have a family of algebras $\{ \mathfrak{A}_t \}_{t \in \mathbb{R}}$ with the algebra $\mathfrak{A}_t$ describing the observables of the environment up to time $t$ (the past) with the isotony condition $\mathfrak{A}_s \subset \mathfrak{A}_t$ whenever $t<s$. We must commit to measuring only \textit{compatible} observables and this restricts us to a filtration $\{ \mathfrak{Y}_t \}_{t \in \mathbb{R}}$  of \textit{commutative} sub-algebras $\mathfrak{Y}_t$ of $\mathfrak{A}_t $.

For concrete quantum stochastic models \cite{HP84,GC85,ParthQSC}, one can derive explicit equations for the filter that are analogous to the classical filtering problem \cite{Belavkin,Slava,BvHJ}. Often, the solution is rendered as a stochastic master equation which is of the form found in many of the quantum trajectory approaches to unravel and simulate master equations.

It should be noted that there are several approximations being made here. In addition to Born-Markov type approximations, we are technically replacing the Planck spectrum with a flat (white noise) spectrum. The form of non-Fock quantum stochastic calculus required does not come with a well-defined number process \cite{HL85}: this is why we consider a quadrature measurement rather than counting quanta of the output!

Our construction will use the Tomita-Takesaki theory in order to describe the commutant of the measurement algebra. There are two commuting representations of thermal Boson fields on the same doubled-up Fock space and both come into play. The second field may seem like a convenient mathematical artifact used to derive the filter, but there are examples where both are physical. An example would be the quantum fields restricted to the left and right Rindler wedges. In the final section, we discuss the problem of unraveling the master equation for a uniformly accelerated observer.

\section{Quantum Stochastic Models}

A quantum probability space $(\mathfrak{A},\mathbb{E})$ consists of a von Neumann algebra $\mathfrak{A}$ and a normal state $\mathbb{E}$. We shall assume that the state is faithful for convenience.

In the case of closed quantum systems described by a finite-dimensional algebra and governed by a Hamiltonian $H$, we have the identity $\mathbb{E} [A \alpha_t (B)] = \mathbb{E} [\alpha_{t+i\beta} (B)A]$ where $\alpha_z (A)= e^{+izH} Ae^{-izH}$ is the Heisenberg evolution map with complex parameter $z$ and the state $\mathbb{E}$ corresponds to the density matrix $\varrho = e^{-\beta H}/Z$, $Z= \text{tr} e^{_\beta H}$. More generally, we say that a family of automorphisms $(\alpha_t)$ on $(\mathfrak{A},\mathbb{E})$ satisfies the $\beta$-KMS condition if $:t \to \mathbb{E} [A \alpha_t (B)] $ on $\mathbb{R}$ can be interpolated continuously to $t+i\beta \mapsto \mathbb{E} [\alpha_{t+i\beta} (B)A]$ on $\mathbb{R}+i \beta$ over the strip $\mathbb{R}+i [0,\beta]$ by a function that is analytic inside the strip.

We recall the celebrated GNS construction with the triple $(\mathfrak{H},\pi, \Phi )$ where $\pi$ is a representation of the algebra $\mathfrak{A}$ as bounded operators on a separable Hilbert space $\mathfrak{H}$ and where $\Psi \in \mathfrak{H}$ is a normalized vector such that $\mathbb{E}[X] = \langle \Psi , \, \pi (X) \Psi \rangle$.
The first step is to view $\mathfrak{A}$ as a vector space with \lq\lq kets\rq\rq  $|A\rangle$ for $A\in \mathfrak{A}$. A sesquilinear form is given by $\langle A|B \rangle \triangleq \mathbb{E}[A^\ast B]$. The vector $\Psi = |I\rangle$ is the distinguished vector, and we obtain a representation $\pi$ of $\mathfrak{A}$ on this vector space by the rule $\pi(X) | A \rangle = | XA \rangle$. We may factor out the null elements and take the Hilbert space completion $\mathfrak{H}$. 

We next recall the modular theory of Tomita and Takesaki \cite{TT}. Starting again with $\mathfrak{A}$ viewed as a vector space, we repeat the completion process but with the alternative sesquilinear form $\langle A|B \rangle' \triangleq \mathbb{E}[BA^\ast ]$. The two processes effectively lead to the same Hilbert space, and we may relate the obtained inner products by $\langle A|B\rangle ' = \langle A | \Delta B\rangle$ where $\Delta \ge0$ is a (typically unbounded) operator called the modular operator. It is convenient to introduce the anti-linear operator $S$ mapping $|A\rangle $ to $|A^\ast \rangle$.
One writes $S= \Delta^{-1/2}J$ where $J$ is an anti-linear involution known as the modular conjugation. The first main result of Tomita-Takesaki Theory is that $J \pi (\mathfrak{A}) J$ is the commutant $\pi(\mathfrak{A}')$. The second is that the maps $\sigma_t (\cdot ) = \pi^{-1} \big( \Delta^{-it} \pi (\cdot ) \Delta^{it} \big)$ form a one-parameter group known as the modular group and they satisfy the KMS condition ($\beta \equiv 1$) on $(\mathfrak{A},\mathbb{E})$.

Following from this, we can give an alternative representation $\pi'$ defined by $\pi' (X) | A\rangle = | AX^\ast \rangle$. Whereas $\pi$ was *-linear, $\pi'$ will be *-anti-linear and so $\pi'(X)^\ast$ will be different from $\pi' (X^\ast)$ in general. It is straightforward to show that the two representations commute:
\begin{eqnarray}
    \pi(X)\pi'(Y) \, |A\rangle = |XAY^\ast \rangle = \pi'(Y) \pi (X)\, |A\rangle .
\end{eqnarray}
In fact, $\pi'(\cdot ) \equiv S \pi(\cdot ) S$. For more details see \cite{Maassen}.

\subsection{The Araki-Woods Representation}
The (Bose) Fock space over a one-particle Hilbert space $\mathfrak{h}$ is the Hilbert space $\Gamma (\mathfrak{h}) = \mathbb{C} \oplus (\bigoplus_{k=1}^\infty \otimes_{\text{symm.}}^k\mathfrak{h})$ where $\otimes_{\text{symm.}}^k\mathfrak{h}$ is the completely symmterized subspace of the $k$-fold tensor product $\otimes^k\mathfrak{h}$. For $f \in \mathfrak{h}$, the exponential vector with test function $f$ is given by
\begin{eqnarray}
    \exp (f) = 1 \oplus f \oplus (\frac{f \otimes f}{\sqrt{2!}}) \oplus (\frac{f \otimes f \otimes f}{\sqrt{3!}}) \oplus \cdots .
\end{eqnarray}
These form a total subset of the Fock space and we have $\langle \exp (f) , \exp (g) \rangle = e^{\langle f,g \rangle}$.

A key feature is the tensor product factorization $\Gamma (\mathfrak{h}_1 \oplus \mathfrak{h}_2) \cong \Gamma (\mathfrak{h}_1) \otimes \Gamma (\mathfrak{h}_2)$ which comes from the natural identification $\exp (f_1 \oplus f_2 )
=\exp (f_1) \otimes \exp (f_2)$.
 
The annihilator with test function $g \in \mathfrak{h}$ is defined by
\begin{eqnarray}
    A(g) \, \exp (f) = \langle g , f \rangle \, \exp (f).
\end{eqnarray}
The map $g\mapsto A(g)$ is anti-linear. Its adjoint $A(\cdot)^\ast$ will be linear. Together, they satisfy the canonical commutation relations $[A(f), A(g)^\ast] = \langle f,g \rangle$.

The vector $\Phi = \exp (0) = 1 \oplus 0 \oplus 0 \oplus \cdots$ is called the Fock vacuum and we have
\begin{eqnarray}
    \langle \Phi , e^{A(g)^\ast - A(g)} \Phi \rangle = 
    e^{- \frac{1}{2} \| g \|^2 }.
\end{eqnarray}
More generally, we say that a state $\mathbb{E}$ is a quasi-free gauge-invariant state whenever we have
\begin{eqnarray}
    \mathbb{E} [e^{A(g)^\ast - A(g)} ] 
    = e^{ - \frac{1}{2} \langle g, (2N+1) g\rangle }
\end{eqnarray}
for $N \ge 0$ a fixed observable on $\mathfrak{h}$. The case $N \equiv 0$ recovers the vacuum field.

The Araki-Woods representation $\pi$ of the boson operators for a quasi-free gauge-invariant state $\mathbb{E}$ corresponds to realizing the operators as operators on a doubled-up Fock space $\mathfrak{H}_1\otimes \mathfrak{H}_2$ with \cite{ArakiWoods}
\begin{eqnarray}
    \pi (A(g)) = A_1 (\sqrt{N+1}g) \otimes I_2 + I_1 \otimes A_2 (\sqrt{N} j \, g)^\ast
\end{eqnarray}
Here $\mathfrak{H}_k$ are copies of $\mathfrak{H}$ and $A_k(\cdot )$ corresponding copies of $A(\cdot )$. We also need to introduce an anti-linear involution $j$.

The appropriate vector is $\Psi= \Phi \otimes \Phi$ and we have
\begin{eqnarray}
    \langle \Psi , e^{\pi(A(g))^\ast - \pi(A(g))} \Psi \rangle
    = e^{ - \frac{1}{2} \langle g, (N+1) g\rangle }e^{ - \frac{1}{2} \langle g, N g\rangle }
    =e^{ - \frac{1}{2} \langle g, (2N+1) g\rangle }
\end{eqnarray}
as desired.

For $N >0$ we have that $\Psi$ is cyclic for the representation; in contrast, the Fock vacuum is annihilated by $A(g)$ when $N=0$.

The Tomita-Takesaki Theory may be applied here with the modular group having the action $\sigma_t (A(g)) = A \big( (\frac{N}{N+1})^{it} g \big)$. In particular, the commuting anti-linear representation is given by
\begin{eqnarray}
    \pi' (A(g)) = A_1 (\sqrt{N}  j \, g) ^\ast \otimes I_2 + I_1 \otimes A_2 (\sqrt{N+1} g)
    .
\end{eqnarray}
\subsection{Quantum Stochastic Calculus}
Taking the one-particle space to be $\mathfrak{h}=L^2 (\mathbb{R}_+,dt)$ one defines the annihilation process to be operators $A_t \triangleq A(1_{[0,t]})$ for $t\ge 0$.

For each $t>0$, we have $L^2 (\mathbb{R}_+,dt) \cong L^2([0,t], dt)\oplus L^2([t, \infty ) , dt)$ and by virtue of the Fock space tensor product factorization property $\Gamma (\mathfrak{h}) \cong \mathfrak{H}_t  \otimes \mathfrak{H}_{[t,\infty )} $ where
$\mathfrak{H}_t  =\Gamma (L^2([0,t], dt)) $ and $ \mathfrak{H}_{[t,\infty )} = \Gamma (L^2([t, \infty ) ,dt )$.

The annihilator $A_t$ (along with its adjoint $A_t^\ast$, the creator) acts trivially on the future factor $\mathfrak{H}_{[t,\infty )}$. Fixing a separable Hilbert space $\mathfrak{h}_0$, called the initial Hilbert space, Hudson and Parthasarathy \cite{HP84,ParthQSC} developed an analogue of the Ito calculus where the consider integrals of the form on $\mathfrak{h}_0 \otimes \Gamma (L^2 (\mathbb{R}_+,dt)) $
\begin{eqnarray}
    \int_0^t ( X_s dA_s +Y_s dA_s^\ast +Z_s ds).
\end{eqnarray}
The integrands are adapted processes: that is, $X_t,Y_t,Z_t$ may act non-trivially on the factor $\mathfrak{h}_0 \otimes \mathfrak{H}_t$ but trivially on the future factor. The increments are understood as future-pointing: that is, informally $dA_t = A_{t+dt} -A_t$ for $dt>0$, etc. They establish the following quantum Ito table:
\begin{eqnarray}
        \centering
        \begin{tabular}{c|cc}
           $ \times$ & $dA$ & $dA^\ast$  \\
            \hline
             $dA$   &  0 & $dt$ \\
             $dA^\ast$ &0 &0
        \end{tabular}
        .
\end{eqnarray}

The general form of a quantum stochastic differential equation for a unitary adapted process $U_t$ on $\mathfrak{h}_0 \otimes \Gamma (L^2 (\mathbb{R}_+,dt)) $ is
\begin{eqnarray}
    dU_t = \bigg\{ L\otimes dA_t^\ast - L^\ast \otimes dA_t
    +K \otimes dt \bigg\} U_t
\end{eqnarray}
with $U_0 = I \otimes I$ and
\begin{eqnarray}
    K=- \frac{1}{2} L^\ast L - i H
\end{eqnarray}
and where $L$ and $H=H^\ast$ are bounded operators on $\mathfrak{h}_0$.

\subsubsection{Non-Fock (Thermal) Noise}
One may generalize this to thermal noise with processes $B_t,B_t^\ast$ satisfying a quantum Ito table ($n>0$) \cite{HL85}:
\begin{eqnarray}
        \centering
        \begin{tabular}{c|cc}
           $ \times$ & $dB$ & $dB^\ast$  \\
            \hline
             $dB$   &  0 & $(n+1)dt$ \\
             $dB^\ast$ & $n\, dt$ & 0
        \end{tabular}
        .
\end{eqnarray}
A unitary adapted process $U_t$ can now be obtained from a quantum stochastic differential equation of the form
\begin{eqnarray}
    dU_t = \bigg\{ L\otimes dB_t^\ast - L^\ast \otimes dB_t
    +K \otimes dt \bigg\} U_t ,
    \label{eq:U_QSDE}
\end{eqnarray}
where now
\begin{eqnarray}
    K=- \frac{1}{2} (n+1) L^\ast L - \frac{1}{2} n LL^\ast - i H.
    \label{eq:K}
\end{eqnarray}
The thermal noise can be easily represented using an Araki-Wood construction:
\begin{eqnarray}
    B_t \equiv \sqrt{n+1} A_{1,t} \otimes I_2 + \sqrt{n} I_1 \otimes A_{2,t}^\ast .
    \label{eq:AW_B}
\end{eqnarray}
Here, $A_{k,t}$ are copies of the usual Fock space annihilation process.
We note that we have the alternative processes
\begin{eqnarray}
    B_t' \equiv \sqrt{n} A_{1,t}^\ast \otimes I_2 + \sqrt{n+1} I_1 \otimes A_{2,t} .
    \label{eq:AW_B'}
\end{eqnarray}
The processes $B_t , B^\ast_t$ commute with the $B_t', B_t^{\prime\ast}$

For $X$ an observable on $\mathfrak{h}_0$, we define its evolution under the unitary quantum stochastic dynamics as
\begin{eqnarray}
    j_t (X) = U_t^\ast (X \otimes I) U_t.
\end{eqnarray}
From the quantum Ito calculus, on has
\begin{eqnarray}
    dj_t (X) = j_t ( [X,L]) \otimes dB^\ast_t
    +j_t ([L^\ast ,X] \otimes dB_t
    +j_t ( \mathcal{L} X ) \otimes dt
\end{eqnarray}
where we encounter the Lindblad generator \cite{Lindblad}
\begin{eqnarray}
    \mathcal{L} X = (n+1) \mathcal{L}_L X+ n \mathcal{L}_{L^\ast} X -i[X,H] .
    \label{eq:L}
\end{eqnarray}
where $\mathcal{L}_M X = \frac{1}{2}  [M^\ast ,X]M + \frac{1}{2} M^\ast [X,M ]$.

\begin{proposition}
The output process is defined as $B_t^{\mathrm{out}}= U^\ast_t (I \otimes B_t ) U_t$ and we have
\begin{eqnarray}
    dB_t^{\mathrm{out}} = dB_t + j_t (L) dt .
    \label{eq:io}
\end{eqnarray}
\end{proposition}
\begin{proof}
    While (\ref{eq:io}) is formally identical to the $n=0$ case, the derivation deserves some attention. Retaining the non-trivial quantum Ito corrections we find
    \begin{eqnarray}
        dB_t^{\text{out}} &=& dB_t + (dU_t^\ast) (I \otimes dB_t) U_t + U_t^\ast (I \otimes dB_t) dU_t
        \nonumber \\
        &=& dB_t - j_t (L) \, n dt +j_t(L) \, (n+1) dt, 
    \end{eqnarray}
    so that $n$-dependent terms cancel leaving the Fock vacuum expression.
\end{proof}

\subsection{Application to the Fulling-Davies-Unruh Effect}

We consider a (classical) observer moving along a world line $\mathscr{W}$ in Minkowski spacetime. The past of the world line is the union $\mathcal{P}_{\mathscr{W}}$ of all past light cones emanating from the observer, that is, all the events in spacetime that can send a signal to the observer. Likewise, the future of the world line is the union $\mathcal{F}_{\mathscr{W}}$ of all future light cones and consists of all the events that can receive a signal from the observer. In the case of an inertial observer both the past and the future will correspond to the entire spacetime. One might suppose that this is true for non-inertial observers, however, this is not the case. The world line $\mathscr{W}$ of a uniformly accelerating observer has the property that there exists a nontrivial region of spacetime which can neither receive signals from the observer nor send signals to the observer. See Figure \ref{fig:4_regions}.

\begin{figure}[h]
    \centering
    \includegraphics[width=0.4\linewidth]{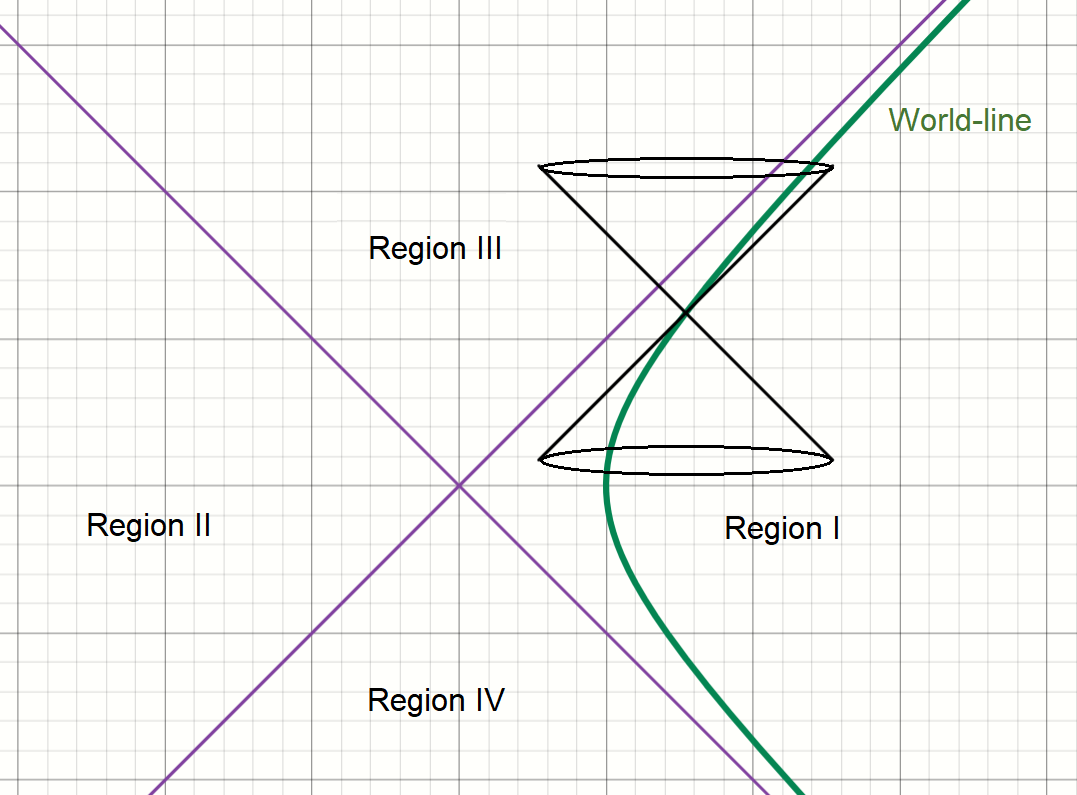}
    \caption{(color online) The world line of a uniformly accelerated observer. Spacetime is divided into four regions. Region III can received signals from the observer but cannot send signals. The situation is reversed for Region IV. Regions I and II are the causal complements to each other.}
    \label{fig:4_regions}
\end{figure}

In Figure \ref{fig:4_regions}, Region I is $\mathcal{P}_{\mathscr{W}}\cap \mathcal{F}_{\mathscr{W}}$ while Region II is $\mathcal{P}_{\mathscr{W}}^\prime\cap \mathcal{F}^\prime_{\mathscr{W}}$. They are causally complementary insofar as the separation between any event in Region I and any event in Region II is always spacelike.

A remarkable feature is the Fulling-Davies-Unruh effect \cite{Fulling}-\cite{Birrell_Davies} which predicts that the uniformly accelerating observer will perceive a vacuum quantum field as a thermal field with temperature $T= \frac{\hbar a}{2 \pi c k_B}$.

From a mathematical perspective, relativistic quantum field theory based on algebras of observables localized in bounded regions are typically type III$_1$ von Neumann algebras rather than the more familiar type I encountered in quantum theory \cite{Yngvason}. The algebraic statements of the Fulling-Davies-Unruh effect and the related Hawking effect can be found in, for example, \cite{Hollands_Wald}.

 The physical interpretation is controversial. There is a consensus that indirect effects (such as thermalization to the Unruh temperature $T$) can be observed; however, whether or not there is radiation that can be directly observed has been the subject of debate. There are good reasons to believe that Unruh radiation itself is not observable \cite{F_OC}. 

In a recent publication \cite{Gough_Entropy2025}, it is shown via a weak coupling limit argument that one may obtain a non-Fock Markovian description of the open dynamics for a system undergoing a uniform acceleration. This yields a quantum stochastic differential equation of the form (\ref{eq:U_QSDE}) in, for instance, the case of an Unruh-DeWitt detector. Here, the detector is a two-level atom  $L= \sqrt{\kappa} \, \sigma_-$ where $\sigma_-$ is the lower operator and which picks up a harmonic dependence $e^{-i\omega t}$ under the free detector dynamics ($H\equiv \omega \, \sigma_+\sigma_-$). The argument can be extended to general $L$ provided the free Heisenberg dynamics is again of the form $e^{-i \omega t} L$. The number of quanta parameter will be the standard bosonic form 
\begin{eqnarray}
    n = \frac{1}{e^{\beta \hbar \omega }-1}
\end{eqnarray}
with $\beta$ the inverse Unruh temperature.

In the following, we will derive the filter corresponding to homodyning the output from an Unruh-DeWitt detector.

\section{Quantum Filtering at Finite Temperature}
We shall consider an open system with Hilbert space $\mathfrak{h}_0$ coupled to a heat bath with quantum stochastic evolution given by the quantum stochastic differential equation (\ref{eq:U_QSDE}). We shall adopt the Araki-Woods representation for the bath processes as described in (\ref{eq:AW_B}). Without loss of generality, we take the state of the system to be the pure state corresponding to a vector $\phi$ and with an abuse of notation just write $\mathbb{E}$ for the separable state where the system is in state $\phi $ and the bath is in the thermal state. The latter then corresponds to the joint Fock vacuum $\Psi = \Phi \otimes \Phi$ in the Araki-Woods representation.

We shall consider measurements of the output quadrature $Y_t = B_t^{\text{out}} +B_t^{\text{out}\ast} $. The von Neumann algebra $\mathfrak{Y}_t$ generated by the observables $\{ Y_s: s\in [0,t]\}$ is commutative. Our aim is to calculate the conditional expectation $\pi_t (X)$ of an observable $j_t (X)$ onto the measurement algebra $\mathfrak{Y}_t$:
\begin{eqnarray}
    \pi_t (X) \triangleq \mathbb{E} [j_t (X)|\mathfrak{Y}_t].
\end{eqnarray}
This conditional expectation is guaranteed to exist since we have the nondemolition property $j_t(X) \in \mathfrak{Y}_t'$ and so one may use the standard spectral theorem construction from \cite{BvH}.

The output quadrature may be written as $Y_t = U_t^\ast  ( I \otimes Z_t) U_t$ where $Z_t = (B_t+B_t^\ast)$. The von Neumann algebra generated by the observables $\{ Z_s: s\in [0,t]\}$ is likewise commutative and is denoted $\mathfrak{Z}_t$. We note that 
\begin{eqnarray}
    dY_t = dZ_t + j_t (L+L^\ast) \, dt
\end{eqnarray}
and that
\begin{eqnarray}
    dY_t \, dY_t = (2n+1) \, dt.
    \label{eq:dY}
\end{eqnarray}

\bigskip

WE now come to a crucial point. Let us note the following nonzero quantum Ito table between the unprimed and primed noises:
\begin{eqnarray}
    dB\, dB'= dB^\ast \, dB^{\prime\ast}=dB'\, dB= dB^{\prime\ast}dB^\ast = \sqrt{(n+1)n}\, dt.
\end{eqnarray}
We introduce the process $Z_t' = \frac{1}{i}(B_t'-B_t^{\prime \ast})$. The $Z$ and $Z'$ processes commute with each other, however, we note the nontrivial Ito table
\begin{eqnarray}
        \centering
        \begin{tabular}{c|cc}
           $ \times$ & $dZ$ & $dZ'$  \\
            \hline
             $dZ$   &  $(2n+1) \, dt$ & $0$ \\
             $dZ'$ & $0$  & $(2n+1) \, dt$
        \end{tabular}
        .
      \label{eq:Z_table}
\end{eqnarray}
In principle, we could have chosen any quadrature $Z'_t= e^{i\lambda}B_t'+e^{-i\lambda }B_t^{\prime\ast}$ so that $Z$ and $Z'$ are commuting (an therefore essentially classical) stochastic processes, but onlt the choice $\lambda= \pi/2$ leads to them being uncorrelated and therefore statistically independent.

\begin{theorem}[Quantum Kallianpur-Striebel]
    The quantum filter is given by
    \begin{eqnarray}
        \pi_t (X) = \frac{\sigma_t (X)}{\sigma_t(I)}
        \label{eq:QKS}
    \end{eqnarray}
    where $\sigma_t (X) = U_t^\ast \mathbb{E}[ V_t^\ast X V_t | \mathfrak{Z}_t ] U_t$ with $V_t \in \mathfrak{Z}_t'$ satisfying the quantum stochastic differential equation
    \begin{eqnarray}
        dV_t = \bigg(\frac{1}{2n+1}\big[ (n+1)L - nL^\ast \big]\otimes dZ_t  - \frac{\sqrt{(n+1)n}}{2n+1} (L+ L^\ast ) \otimes dZ'_t 
         + K \otimes dt \bigg) V_t ,
        \label{eq:dV}
    \end{eqnarray}
    with $V_0 = I\otimes I$.
\end{theorem}

\begin{proof}
    The form (\ref{eq:QKS}) follows from \cite{BvH} Lemma 6.1 provided we can construct $V_t \in \mathfrak{Z}_t'$ with the property that $\mathbb{E}[j_t(X)]= \mathbb{E}[V_t^\ast (X \otimes) V_t]$. We now give this construction.

    Employing the Araki-Woods representation, the overall state comes from the vector state $\phi \otimes \Psi$ and we note that
    \begin{eqnarray}
        dU_t \, \phi \otimes \Psi \equiv
        \bigg( \sqrt{n+1}L\otimes dA_{1,t}^\ast - \sqrt{n} L^\ast \otimes dA_{2,t}^\ast +K \otimes dt \bigg) U_t \, \phi \otimes \Psi,
        \label{eq:QSDE_AW}
    \end{eqnarray}
with $K$ given by (\ref{eq:K}). The annihilator differentials $dA_{k,t}$ automatically disappear as they act on the corresponding vacuum. We now note that
\begin{eqnarray}
    dZ_t \, \Psi &=& \sqrt{n+1}\, dA_{1,t}^\ast \Psi+ \sqrt{n}\, dA_{2,t}^\ast\Psi , \nonumber \\
    dZ_t' \,\Psi &=& \frac{1}{i}\sqrt{n}\, dA_{1,t}^\ast \Psi - \frac{1}{i} \sqrt{n+1}\, dA_{2,t}^\ast\Psi ,
\end{eqnarray}
which may be inverted to give
\begin{eqnarray}
    dA_{1,t}^\ast \, \Psi &=& \frac{\sqrt{n+1}}{2n+1}\, dZ_t \, \Psi +i \frac{ \sqrt{n}}{2n+1} \,dZ'_t \, \Psi \nonumber \\
    dA_{2,t}^\ast \, \Psi &=&  \frac{\sqrt{n}}{2n+1}\, dZ_t \, \Psi -i \frac{ \sqrt{n+1}}{2n+1} \,dZ'_t \, \Psi .
\end{eqnarray}
Making these replacements in (\ref{eq:QSDE_AW}) leads to
\begin{eqnarray}
        dU_t \, \phi \otimes \Psi &\equiv&
        \bigg( \big( \frac{n+1}{2n+1}L - \frac{n}{2n+1}L^\ast \big) \otimes dZ_t  +i \frac{ \sqrt{n(n+1)}}{2n+1}\, (L+ L^\ast ) \otimes dZ'_t  \nonumber \\
        &&\qquad \qquad \qquad +K \otimes dt \bigg) U_t \, \phi \otimes \Psi .
        \label{eq:QSDE_Z}
    \end{eqnarray}
It is now clear that $\mathbb{E}[j_t (X)] = \langle \phi \otimes \Psi , U_t^\ast (X \otimes I ) U_t \, \phi \otimes \Psi \rangle \equiv \langle \phi \otimes \Psi , V_t^\ast (X \otimes I ) V_t \, \phi \otimes \Psi \rangle$, as desired.
\end{proof}

\begin{corollary}(Quantum Zakai Equation)
\label{Cor_QZE}
    The unnormalized filter $\sigma_t (X) $ satisfies
    \begin{eqnarray}
        d\sigma_t (X) = \sigma_t (\mathcal{L} X) \, dt
        + \frac{1}{2n+1}\bigg\{ (n+1) \sigma_t (XL+L^\ast X) -n \, \sigma_t (XL^\ast + L X ) \bigg\} dY_t .
        \label{eq:qZakai}
    \end{eqnarray}
    Here, $\mathcal{L}$ is the Lindblad generator from (\ref{eq:L}).
\end{corollary}
\begin{proof}
    Using the table (\ref{eq:Z_table}), we find that
    \begin{eqnarray}
        V_t^\ast (X \otimes I) V_t &=& X\otimes I+ \frac{1}{2n+1} \int_0^t
        V_s^\ast \big[ X \big( (n+1) L +nL^\ast )\otimes I \big]V_s \otimes dZ_s \nonumber \\
        &+& \frac{1}{2n+1}\int_0^t
        V_s^\ast \big[ \big( (n+1) L^\ast +nL )X \otimes I \big]V_s \otimes dZ_s \nonumber \\
        &+&  \frac{\sqrt{(n+1)n}}{2n+1} \int_0^t
        V_s^\ast \big( \big[ X (L+ L ^\ast ) + (L+L^\ast  )X \big]\otimes I \big)V_s \otimes dZ'_s \nonumber \\
        && + \int_0^t V_s \big( \mathcal{M} X \big) \otimes I \big) V_s\otimes ds
    \end{eqnarray}
    where
    \begin{eqnarray}
        \mathcal{M} X &=& K^\ast X+XK \nonumber \\
        & +& \frac{1}{2n+1} \bigg( (n+1) L^\ast -nL \bigg) X\bigg( (n+1) L -nL^\ast \bigg) + \frac{n(n+1)}{2n+1} (L+L^\ast) X ( L+L^\ast)  .
    \end{eqnarray}
    After a routine expansion and collection of terms, one finds that $\mathcal{M}$ is, in fact, reduces to $\mathcal{L}$.

    We now take the conditional expectation onto $\mathfrak{Z}_t$ to obtain
    \begin{eqnarray}
        \mathbb{E} [ V_t^\ast (X \otimes I) V_t  | \mathfrak{Z}_t] 
        &=& X\otimes I \nonumber +\frac{1}{2n+1} \int_0^t
        \mathbb{E} \big[ 
        V_s^\ast \big[ X \big( (n+1) L +nL^\ast )\otimes I \big]V_s | \mathfrak{Z_s}\big] \otimes dZ_s \nonumber \\
        && -\frac{1}{2n+1}\int_0^t
        \mathbb{E} \big[ 
        V_s^\ast \big[ \big( (n+1) L^\ast  +nL )X\otimes I \big]V_s | \mathfrak{Z_s}\big] \otimes dZ_s \nonumber \\
        &&+  \int_0^t
        \mathbb{E} \big[ 
        V_s^\ast \big[ \mathcal{L}X \otimes I \big]V_s | \mathfrak{Z_s}\big] \, ds
    \end{eqnarray}
    using the fact that the increments $dZ_t'$ will have conditional mean zero.
    This implies that
    \begin{eqnarray}
        d\,  \mathbb{E} [ V_t^\ast (X \otimes I) V_t  | \mathfrak{Z}_t] 
        &=& \frac{1}{2n+1} \mathbb{E} \big[ 
        V_t^\ast \big[ X \big( (n+1) L -nL^\ast )\otimes I \big]V_t | \mathfrak{Z_t}\big] \otimes dZ_t \nonumber \\
        &&-\frac{1}{2n+1}\mathbb{E} \big[ 
        V_t^\ast \big[ \big( (n+1) L^\ast  -nL )\otimes I \big]V_t | \mathfrak{Z_t}\big] \otimes dZ_t \nonumber \\
        &&+  
        \mathbb{E} \big[ 
        V_t^\ast \big[ \mathcal{L}X \otimes I \big]V_t | \mathfrak{Z_t}\big] \, dt .
    \end{eqnarray}
    It only remains to perform the unitary conjugation to get $\sigma_t (X) = U_t^\ast \mathbb{E} [ V_t^\ast (X \otimes I) V_t  | \mathfrak{Z}_t] U_t$ and note that this conjugation converts $Z_t$ into $Y_t$.
\end{proof}

\begin{corollary}(Quantum Stratonovich-Kushner Equation)
    The normalized filter $\pi_t (X)$ satisfies the equation
    \begin{eqnarray}
        d \pi_t (X) &=& \pi_t (\mathcal{L }X) \, dt
        +\bigg\{ \frac{n+1}{2n+1} \big[ \pi_t (XL+L^\ast X) - \pi_t (X) \pi_t (L+L^\ast ) \big] \nonumber \\
        && - \frac{n}{2n+1} \big[ \pi_t (XL^\ast +L X) - \pi_t (X) \pi_t (L+L^\ast ) \big]
        \bigg\} \, dI_t
        \label{eq:SK}
    \end{eqnarray}
    where the innovations process $I_t$ is defined by $I_0=0$ and
    \begin{eqnarray}
        dI_t = dY_t - \pi_t (L+L^\ast )\, dt .
    \label{eq:innovations}
    \end{eqnarray}
\end{corollary}

\begin{proof}
    We first observe that the normalization factor satisfies $d\sigma_t (I) = \frac{1}{2n+1} \sigma_t (L+L^\ast ) \, dY_t$. From the Ito formula and (\ref{eq:dY}), we see that
    \begin{eqnarray}
        d \frac{1}{\sigma_t (I)} = -\frac{1}{2n+1}\frac{\sigma_t (L+L^\ast ) }{\sigma_t (I)^2}  \, dY_t + \frac{1}{2n+1}\frac{ \sigma_t (L+L^\ast)^2 }{\sigma_t (I)^3} \, dt ,
    \end{eqnarray}
    therefore, from the Ito product rule
    \begin{eqnarray}
        d \pi_t (X) = d \frac{\sigma_t (X)}{\sigma _t (I)}
        = d\sigma_t (X) \, \frac{1}{\sigma _t (I)}
        +\sigma_t (X)\, d\frac{1}{\sigma _t (I)}
        +d\sigma_t (X)\, d\frac{1}{\sigma _t (I)}
    \end{eqnarray}
    and the desired result (\ref{eq:SK}) readily follows.
\end{proof}

\begin{remark}
    In the zero temperature case, we may present the filter as a vector state valued process. For instance, Let $|\chi_t \rangle$ satisfy the stochastic differential equation $|d \chi_t \rangle =-( \frac{1}{2}L^\ast L +iH) | \chi_t \rangle + | \chi_t \rangle \, dY_t$ with $|\chi_0\rangle=|\psi_0 \rangle$ a fixed normalized vector. Then $\sigma_t (X)=\langle \chi_t |X | \chi_t\rangle$ satisfies the quantum Zakai equation (\ref{eq:qZakai}) with $n=0$ and initial condition $\sigma_0 (X) = \langle \psi_0 |X| \psi_0 \rangle$. Likewise, $\| \chi_t \|^2 \equiv \sigma_t (I)$ and we may write the filter as $\pi_t (X) = \langle \psi_t |X| \psi_t \rangle$ where $\psi_t = \chi_t /\|\chi_t \|$. This does not naturally extend to the $n>0$ case, however. Instead, we must start with (\ref{eq:QSDE_Z}) which leads to 
    \begin{eqnarray}
         | d \Xi_t \rangle =
        \bigg( \frac{1}{2n+1}\big[ (n+1)L - nL^\ast \big]\otimes dY_t  - \frac{\sqrt{(n+1)n}}{2n+1}\, (L+ L^\ast ) \otimes dY'_t  +K \otimes dt \bigg) | \Xi_t \rangle.
        \label{eq:unravel}
    \end{eqnarray}
    Note that $\Xi_t \rangle$ is adapted to both the output $Y_t$ and additionally to $Y_t'\equiv U_t^\ast [ 1\otimes Z_t']U_t$. The increments $dY_t , dY_t'$ satisfy the same Ito table (\ref{eq:Z_table}) as $dZ_t,dZ_t'$.
    From the Ito calculus, $d\langle\Xi_t |X| \Xi_t \rangle =\langle\Xi_t |X| d \Xi_t \rangle+\langle d\Xi_t |X| \Xi_t \rangle+\langle d\Xi_t |X| d\Xi_t \rangle$ and using the calculations in Corollary \ref{Cor_QZE} for the quantum Zakai equation we find
    \begin{eqnarray}
        d\langle\Xi_t |X| \Xi_t \rangle &=&\langle\Xi_t |\mathscr{L} X|  \Xi_t \rangle dt
        +\frac{1}{2n+1}\langle \Xi_t |\big[ (n+1)L - nL^\ast \big]^\ast X+X\big[ (n+1)L - nL^\ast \big]| \Xi_t \rangle dY_t
        \nonumber \\
        && - \frac{\sqrt{(n+1)n}}{2n+1}\, \langle \Xi_t |(L+L^\ast)X+X(L+L^\ast) | \Xi_t \rangle dY_t' .
    \end{eqnarray}
    The (Zakai) unnormalized filter $\sigma_t (X)$ is then obtained by taking the conditional expectation of $\langle\Xi_t |X| \Xi_t \rangle$ with respect to $\mathfrak{Y}_t$. In this sense, the unraveling is by $| \Xi _t\rangle$ but we need a further conditioning so that we coarse-grain over the measurements and average out the $Y'_t$ process.
\end{remark}

\section{Acknowledgement}
I have the pleasant duty to thank Luc Bouten for several stimulating discussions on quantum filtering, especially with regard to the reference probability approach.

\end{document}